\documentclass[sigconf]{aamas}

\usepackage{balance} % for balancing columns on the final page

\usepackage{thm-restate}
\usepackage{todonotes}
\usepackage[utf8]{inputenc}
\usepackage{algorithm}
\usepackage{algpseudocode}
\usepackage{amsfonts,amsmath}
\usepackage[inline]{enumitem}
\usepackage{nicefrac,xfrac}
\usepackage{fancyhdr}
\setcopyright{ifaamas}
\acmConference[AAMAS '24]{Proc.\@ of the 23rd International Conference
on Autonomous Agents and Multiagent Systems (AAMAS 2024)}{May 6 -- 10, 2024}
{Auckland, New Zealand}{N.~Alechina, V.~Dignum, M.~Dastani, J.S.~Sichman (eds.)}
\copyrightyear{2024}
\acmYear{2024}
\acmDOI{}
\acmPrice{}
\acmISBN{}

\title{Gerrymandering Planar Graphs}

\author{Jack Dippel}
\affiliation{
  \institution{McGill University}
  \city{}
  \country{}
  }
\email{}
\author{Max Dupré la Tour}
\affiliation{
  \institution{McGill University}
  \city{}
  \country{}
  }
  \author{April Niu}
\affiliation{
  \institution{McGill University}
  \city{}
  \country{}
  }
  \author{Sanjukta Roy}
\affiliation{
  \institution{Penn State University}
  \city{}
  \country{}
  }
\author{Adrian Vetta}
\affiliation{
  \institution{McGill University}
  \city{}
  \country{}
  }
\begin{abstract}
We study the computational complexity of the map redistricting problem (gerrymandering). Mathematically, the electoral district designer (gerrymanderer) attempts to partition a weighted graph into $k$ connected components (districts) such that its candidate (party) wins as many districts as possible. Prior work has principally concerned the special cases where the graph is a path or a tree. Our focus concerns the realistic case where the graph is planar. We prove that the gerrymandering problem is solvable in polynomial time in $\lambda$-outerplanar graphs, when the number of candidates and $\lambda$ are constants and the vertex weights (voting weights) are polynomially bounded. In contrast, the problem is NP-complete in general planar graphs even with just two candidates. This motivates the study of approximation algorithms for gerrymandering planar graphs. However, when the number of candidates is large, we prove it is hard to distinguish between instances where the gerrymanderer cannot win a single district and instances where the gerrymanderer can win at least one district. This immediately implies that the redistricting problem is inapproximable in polynomial time in planar graphs, unless P=NP. This conclusion appears terminal for the design of good approximation algorithms -- but it is not. The inapproximability bound can be circumvented as it only applies when the maximum number of districts the gerrymanderer can win is extremely small, say one. Indeed, for a fixed number of candidates, our main result is that there is a constant factor approximation algorithm for redistricting unweighted planar graphs, provided the optimal value is a large enough constant. 
\end{abstract}

\newcommand{\tw}{tw}

\begin{document}

%%% The following commands remove the headers in your paper. For final 
%%% papers, these will be inserted during the pagination process.

\pagestyle{fancy}
\fancyhead{}

%%% The next command prints the information defined in the preamble.

\maketitle 
\section{Introduction}\label{sec:intro}
Partisan gerrymandering refers to the manipulation of district lines in order to give an advantage to one political party over others. This is facilitated by the hierarchical structure of political elections which allows for manipulation at the district or division level to ultimately secure more seats.
An effective strategy that divides the voting population can lead to a party or candidate winning a higher number of seats than they would otherwise have won. This process is commonly referred to as {\em gerrymandering}. The US House of Representatives elections provide a clear example of this practice, where candidates from different parties compete at the district level to represent their constituents in Congress. The problem of gerrymandering has recently gained considerable attention in the computer science community. Algorithms are used in practice to create districts 
(for example, many states within Mexico~\cite{trelles2014effects} use computer algorithms to district) 
and also to test for gerrymandering. The discussion of algorithmically finding optimal districts started in the ’60s~\cite{hess1965nonpartisan,altman1997computational}.
%
% While social scientists have made strides in presenting the Supreme Court with effective ways of measuring partisan asymmetry in the conversion of votes to seats, a crucial issue remains: how can the courts differentiate between acceptable districting efforts by state legislatures, which involve the application of conventional criteria, the preservation of communities of interest, and the promotion of minority representation, and partisan gerrymandering?
%If the degree of asymmetry in the conversion of votes to seats is highly pronounced, it becomes easier to detect instances of gerrymandering.
% \todo[inline]{SR: What do you feel about keeping or not keeping the this degree-asymmetry motivation here? I talk about this again in our contribution to motivate why we are looking at singleton winning districts. Should I move it there?}
This induces the algorithmic question of whether it is possible to redistrict in such a way as to secure a greater number of seats for the favored political party (the {\em gerrymanderer}). 

A ubiquitous constraint in gerrymandering is that the counties that form a district be
contiguous. Thus gerrymandering is simply a graphical problem, where a vertex represents a county and an edge denotes the existence of a shared border 
between the two corresponding counties. The task then is to partition the graph into connected subgraphs, namely districts, such that one specific candidate (or party) will win as many districts as possible. 
This problem, called {\em gerrymandering over graphs}
was introduced by Cohen-Zemach, Lewenberg, and Rosenschein~\cite{ZemachLR18}.
The computational complexity of the gerrymandering problem has since received considerable attention in the literature. Clearly, planar graphs capture the setting of interest in practice, namely redistricting a geographical area. However, prior work on the gerrymandering problem has predominantly focused upon the special case of networks that are either paths or trees \cite{ZemachLR18,ItoKKO19,gupta2021gerrymandering,bentert2021complexity,SCD20,BKN22}.
But what can we say about real instances, specifically, what happens in maps? Answering this question motivates our study of gerrymandering on {\em planar} graphs. 

% Ito et al.~\cite{ItoKKO19} further extended this study to various classes of graphs, such as paths, trees, complete bipartite graphs, and complete graphs. Both~\cite{ZemachLR18,ItoKKO19} explore a hierarchical voting process that involves dividing a set of voters into multiple districts, each of which conducts a separate election. Within each district, one candidate is selected as the nominee using the plurality rule. Finally, the elected nominees compete against each other using the same plurality rule to determine the overall winner.

\subsection{The Gerrymandering Model}
% An election is comprised of a set $V$ of \textit{voters} and a set $\mathcal{C}$ of \textit{candidates}. 
An election with a set $\mathcal{C}$ of $c$ candidates is modeled by a graph $G=(V,E)$ where vertices represent counties and edges
indicate geographic adjacency.
The {\em gerrymandering} problem of Cohen-Zemach et al.~\cite{ZemachLR18} is then as follows.
There is a vertex weight function $wt \colon V(G) \rightarrow \mathbb{N}$
and an approval function $a:V(G) \rightarrow \mathcal{C}$, where $a(v)$ represents 
the candidate approved of by $v$.
Given a distinguished candidate $i\in \mathcal{C}$ and a connected subgraph
$T\subseteq V(G)$, we say that candidate \emph{$i$ wins district $T$} if
$\sum_{v\in T, \,a(v)=i} wt(v) > \sum_{v\in T, \,a(v)=j} wt(v)$, for any other candidate $j$. 
\begin{quote}
\textsc{The Gerrymandering Problem:} 
Given two positive integers $k$ and $w$, and a candidate $b$, is there a
partition of $V(G)$ into $k$ non-empty, connected districts $T_1\uplus \ldots \uplus T_k$ such that candidate $b$ wins at least $w$ districts?
\end{quote}
We will also consider a \emph{bounded district version} of the problem where we incorporate lower and upper bounds on the number of counties in any district (or lower and upper bounds on the total weight of the vertices of any district).
That is, $\ell \le |T_i| \le u$ for some $1\leq \ell \leq u \leq |V(G)|$ and each district $T_i$. 

A feasible solution is called a $(k,w)$-\textit{partition}.
The decision problem is to determine whether a $(k,w)$-partition
exists. The corresponding optimization problem is to find the largest
possible $w_{OPT}$ such that a $(k,w_{OPT})$-partition exists.

Alternately, we can represent each candidate by a color, with the gerrymanderer $b$ being denoted by ``blue". Each vertex is given the color of its preferred candidate.
Given a partition of the graph into $k$ connected components (districts),
blue wins a district if it has strictly the highest total weight of any color in the
district. The gerrymanderer desires a partition that maximizes 
the number of districts won by blue. We say a district is a 
\emph{winning district} if blue wins it.

\subsection{Our Contribution}
% The vertices of the graph represents the voters in a region or a county, the weight of the vertex can be the size of this region, and the edges model their geographic closeness or connection. 
Our focus is upon the natural restriction that the graph $G$ is planar.
However, Ito et al.~\cite{ItoKKO19} showed the gerrymandering problem is 
NP-complete, even if the number of districts is $k=2$, the number of candidates is $c=2$, and $G$ is the bipartite graph $K_{2,n}$. Since $K_{2,n}$ is planar, the problem is hard for planar graphs. 
In this paper we present much stronger hardness results and initiate the study of approximation algorithms for the gerrymandering problem in planar graphs. Moreover, we provide exact algorithms for a slight restriction of planar graphs.

\textbf{Exact solutions.} Redistricting graphs is shown to be NP-hard even in restricted graph classes~\cite{bentert2021complexity,Fraser2022,gupta2021gerrymandering} and there have been a few algorithmic results, mainly for simple structures such as paths~\cite{gupta2021gerrymandering} and trees~\cite{bentert2021complexity,Fraser2022}. We push the boundaries of known complexity results. In Section~\ref{section:boundedTW}, we study $\lambda$-outerplanar graphs.
A graph $G$ is \emph{$\lambda$-outerplanar} if $G$ has a planar embedding such that the vertices belong to $\lambda$ layers.
Define $L_1$ to be the vertices incident to the outer-face, and define $L_i$ for $i>1$ recursively to be the vertices on the outer-face of the planar drawing obtained by removing the vertices in $L_1,L_2, \dots, L_{i-1}$. Each $L_i$ is called a layer.\footnote{Equivalently, iteratively deleting the outer layer (unbounded face) of the graph will produce an empty graph in $\lambda$ iterations.} We remark that $\lambda$-outerplanar graphs with a constant index $\lambda$ are important
as the graphs that arise in practical maps have this property.
We prove that, for constant~$\lambda$ and a constant number of candidates, there
is a polynomial time algorithm to solve the gerrymandering problem for $\lambda$-outerplanar graphs.
This positive result (Theorem~\ref{thm:outerplanar}) applies to the districting model in its full generality, including with lower and upper bounds on the district sizes, provided each vertex weight is polynomially bounded. Note that this result is tight; that is, unless P=NP, we cannot remove the dependency on the weight in the running time since it is NP-hard for paths with general weights~\cite{bentert2021complexity}.
To show Theorem~\ref{thm:outerplanar} we, in fact, prove a more general result (Theorem~\ref{thm:boundedTW}): the gerrymandering problem is 
solvable in polynomial time in graphs of bounded treewidth, given a constant number of
candidates and polynomial vertex weights. In particular, our result extends the approach of~\cite{ItoKKO19} on trees to graphs of bounded treewidth.
%\todo{SR:Need to cite something?}

\textbf{Hardness of approximation.} In Section~\ref{sec:hardness} we explore the hardness of the gerrymandering problem in more detail.
First, we consider the case of two candidates and general graphs.
We show, via a reduction from the independent set problem, that the gerrymandering problem is inapproximable to within an $O(n^{\frac13-\epsilon})$ factor, for any constant $\epsilon>0$ (Theorem~\ref{thm:hard-general}). This result holds even for the {\em unweighted} case where each vertex has weight one.
We remark this result only implies NP-hardness for planar graphs and not inapproximability because there exists a polynomial-time approximation scheme (PTAS) for independent set in planar graphs \cite{Baker1994}.
Much stronger results arise with many candidates. Specifically, for planar graphs
it is hard to distinguish between instances where the gerrymanderer 
cannot win a single district and instances where the gerrymanderer can win 
at least one district (Theorem~\ref{thm:hard-planar}). This immediately implies that 
the gerrymandering problem is inapproximable in polynomial time in planar graphs, even in the unweighted case!

\textbf{Approximate solutions.} This hardness result suggests no approximation algorithm is possible for planar graphs.
But the situation is more subtle. It actually implies no approximation is possible if
the optimal solution $w_{OPT}$ is small, for example if $w_{OPT}=1$.
Remarkably, in Section~\ref{sec:constant-approx}, we show that for a constant number of candidates there is a constant factor
approximation algorithm in unweighted planar graphs when $w_{OPT}$ is a large enough constant.
Specifically, we present a quasi-linear time $O(c)$-approximation algorithm provided the optimal number of winning districts is a sufficient multiple of the number of 
candidates~$c$ (Theorem~\ref{thm:planarapproxlinear}). 

In Section~\ref{sec:ptas}, we study a problem called the {\em singleton winning district gerrymandering problem} where
the gerrymanderer aims to win only singleton districts. 
This combinatorially interesting but seemingly contrived problem is significant because its approximability
relates closely to the approximability of the generic gerrymandering problem.
Combining our algorithm for $\lambda$-outerplanar graphs with Baker's method~\cite{Baker1994} for planar graphs,  
we prove that there is a polynomial time approximation scheme (PTAS) for the 
unweighted singleton winning district gerrymandering problem in planar graphs (Theorem~\ref{thm:ptas}).
This raises an intriguing structural question that we leave open: does a PTAS also exist for the (non-singleton winning) gerrymandering problem, again assuming $w_{OPT}$ is large enough?

\subsection{Background and Related Work}
For many decades, the existence of gerrymandering and its consequences have been widely acknowledged and discussed in the realm of political science, as documented by Erikson~\cite{Erikson72}, Lubin~\cite{Lub99}, and Issacharoff~\cite{Issacharoff02}. However, the practical feasibility and broad implications of gerrymandering have only recently become the subject of intense public, policy, and legal debate~\cite{NYT-GM2}, largely due to the widespread use of computer modeling in the election process. %As a result, gerrymandering has become the focal point for redrawing political district boundaries and shaping political battle lines.

Computational complexity studies of gerrymandering have taken several forms.  Puppe and Tasn{\'{a}}di~\cite{PUPPE200993} explored gerrymandering under certain constraints where specific groups of voters cannot be included in the partition and showed it is NP-complete. Fleiner et al.~\cite{FleinerNT17}, Lewenberg et al.~\cite{LewenbergLR17} and Eiben et al.~\cite{EFPS20}  investigated gerrymandering in presence of geographical constraints.
%Fleiner et al.~\cite{FleinerNT17}  investigated gerrymandering while considering geographical constraints, requiring each group to be formed by a simply connected region in the plane. They also proved that this problem is NP-complete. On the other hand, Lewenberg et al.~\cite{LewenbergLR17} and Eiben et al.~\cite{EFPS20} focused on studying gerrymandering in a geographical context, where voters are required to vote in the nearest polling stations, making the problem about the strategic placement of polling stations rather than the drawing of district lines.

Most closely related to our work, there has been extensive recent study on
gerrymandering graphs. The model was originally introduced by Cohen-Zemach et al. \cite{ZemachLR18}. They showed it is NP-complete to decide if there is a $(k,w)$-partition of a weighted graph, in the case where districts of size $1$ are not allowed. They also designed a greedy algorithm for gerrymandering over graphs and analyzed its performance empirically through simulations on random graphs.
Ito et al. \cite{ItoKKO19} extended those results by considering a slightly different model, where the goal is to win more districts than any other candidate in a global election. We refer to this model as the \textit{global election model}. They proved that the problem is NP-complete for several graph classes: complete graphs, $K_{2,n}$, unweighted graphs with $c=4$, and trees of diameter four. They also provided polynomial-time algorithms for stars and pseudo-polynomial time algorithms for paths and trees when $c$ is constant.
Gupta et al. \cite{gupta2021gerrymandering} extended the model of gerrymandering to include \textit{vector weighted vertices}, where a vertex can have different weights for different candidates. They proved, in the global election model, that the problem of gerrymandering is NP-complete on paths. They also presented algorithms with running times of $2^k(n+c)^{O(1)}$ and $2^n(n+c)^{O(1)}$ for paths and general graphs, respectively.
Bentert et al. \cite{bentert2021complexity} considered the global election model with trees only and proved NP-hardness for the problem on paths, even in the unweighted case. They also provided a polynomial-time algorithm for trees when $c=2$ and proved weakly NP-hardness for trees when $c\geq 3$. They further presented a polynomial-time algorithm for trees of diameter exactly three.
Furthermore, Fraser et al. \cite{Fraser2022} considered the global election model on trees. They proved that the problem is W[2]-hard when parameterized by the number of districts~$k$, for trees of depth two. They also provided an algorithm with running time $O(n^{3l}\cdot 2.619^k(n + m)^{O(1)})$ for the vector weighted case on trees where $l$ is the number of leaves.
Graph theoretic formulation of districting problem has been used in~\cite{apollonio2009bicolored,borodin2018big} to study structural properties of the problem on grids.

Gerrymandering over graphs has also been studied from the complementary viewpoint of fairness as opposed to gerrymandering.
Stoica et al.~\cite{SCD20} studied gerrymandering graphs with an objective of creating fair connected districts: where the maximum margin of victory of any candidate is minimized, and showed NP-hardness for $k=2$ and $c=2$.  Boehmer et al.~\cite{BKN22} showed W[1]-hardness and an XP algorithm parameterized by treewidth, $k$ and $c$. We emphasize our algorithm for gerrymandering in graphs of bounded treewidth is independent of~$k$.

% \section{Preliminaries} 
% \noindent{\bf Graphs.} A graph $G$ is said to be {\em connected} if every two vertices of $G$ are connected to each other by a path in $G$. A graph $G$ is \emph{$\lambda$-outerplanar} if $G$ has a planar embedding such that the vertices belong to $\lambda$ layers. Due to paucity of space we refer the readers to \cite{DBLP:books/daglib/0030488} for standard graph theoretic definitions.

\section{A Polynomial Time Algorithm for Graphs with Bounded Treewidth}\label{section:boundedTW}
In this section, we introduce a dynamic program approach that efficiently computes an exact solution in polynomial time, provided that the treewidth of the graph and the number of candidates are constant. It is motivated by algorithm for trees of Ito et al.~\cite{ItoKKO19}. Here, the graph is weighted, that is, each vertex represents multiple votes for a candidate. Additionally, it allows for other constraints such as incorporating lower and upper bounds on district sizes.

\begin{restatable}{theorem}{boundedTW}\label{thm:boundedTW}
    There exists an algorithm computing the maximum number of winning district in time $  O(n^{2\tw+7}\cdot (\sum\limits_{v\in V} wt(v))^{2c\cdot\tw})$, where $\tw$ is the treewidth of the graph.
\end{restatable}

We defer the proof of Theorem~\ref{thm:boundedTW} and technical definitions required for the proof to Appendix~\ref{sec:missingDPproof}.
We remark that Theorem~\ref{thm:boundedTW} is tight in the sense that we cannot eliminate the dependency on $c$.
This is because the problem is NP-complete for paths (whose treewidth is one), a result due to Bentert et al.~\cite{bentert2021complexity}.\footnote{The result in \cite{bentert2021complexity} is for the global election model but the reduction applies for the model of~\cite{ZemachLR18} studied~here.} 
%{\bf is this correct?}}. %: If the number of candidates is unbounded, the problem is $NP$-complete for paths. 
Moreover, it is impossible to omit the vertex-weight term and obtain an algorithm with running time  $n^{O(c\cdot\tw)}$,  unless P=NP, since the problem is weakly NP-hard even if $\tw=1$ (namely, trees) and $c=3$~\cite{bentert2021complexity}. 

We further remark that if the graph $G$ is $\lambda$-outerplanar, that is, $G$ has a planar embedding such that the vertices belong to $\lambda$ layers, then $G$ has treewidth at most $3\lambda-1$; see Bodlaender~\cite{Bodlaender1998}. 
Thus $\lambda$-outerplanar graphs have bounded treewidth if $\lambda$ is constant. 
So our dynamic program can be applied to the practical setting of $\lambda$-outerplanar graphs. Specifically, we obtain the following result as a corollary of Theorem~\ref{thm:boundedTW}.
\begin{theorem}\label{thm:outerplanar}
    In a $\lambda$-outerplanar graph $G$, there is an algorithm to compute the maximum number of winning districts in time $n^{O(c\lambda)}$ if the vertex weights are polynomially bounded.
\end{theorem}

\section{On the Hardness of the gerrymandering Problem}\label{sec:hardness}

% {\bf [NP-hard: unweighted paths with many candidates...]}

It is known that gerrymandering is NP-hard for unweighted paths when there are many candidates \cite{bentert2021complexity}.
%,\cite[Theorem 3.1]{ItoKKO19}.{\bf [Don't some of these results apply in the unweighted setting?...][SR: they do but not for planar, I updated the sentence]}
In this section, we present much stronger inapproximability results for unweighted (planar) graphs. %Furthermore, these results apply even in the setting of unweighted graphs, i.e., each vertex has weight one.

\subsection{Lower Bounds for Elections with Two Candidates}%\todo{SR: we never really said what we call as an election}
To begin, we consider the natural setting where there are two candidates (or parties) in the election. Here we show that the gerrymandering problem is at least as hard as the {\em maximum independent set problem}. Proofs for the next two results can be found in Appendix~\ref{sec:missinghardnessproof}.
\begin{restatable}{lemma}{hardness}\label{lem:hardness}
Given a graph $G$ on $n$ vertices and an integer $w$, there exists a graph~$H$ of size polynomial in $n$ and a coloration of $H$ using two colors 
such that there exists a $w$-sized independent set in $G$ {\em if and only if} there exists a $(n,w)$-partition in $H$.
\end{restatable}
Using this link to the independent set problem~\cite{Hstad1999,Z2006} we can obtain a very strong inapproximability bound for the gerrymandering problem with two candidates in general~graphs. 
\begin{restatable}{theorem}{hardgeneral}\label{thm:hard-general}
In gerrymandering with two candidates, for any $\epsilon>0$, 
there is no polynomial-time algorithm with approximation guarantee less than $n^{1/3-\epsilon}$, unless P = NP.
\end{restatable}
% \begin{proof}[Proof sketch, see Appendix~\ref{sec:missinghardnessproof}]
% The idea is to use the well known inapproximability of maximum independent set proven by Zuckerman \cite{Z2006} and apply Lemma~\ref{lem:hardness}.
%  \end{proof}

\subsection{Lower Bounds for Multi-Candidate Elections in Planar Graphs}
We remark that, for a planar graph $G$, the reduction in Lemma~\ref{lem:hardness} produces a planar auxiliary 
graph $H$. Thus, Lemma~\ref{lem:hardness} implies that, for two candidates the
gerrymandering problem is NP-complete in planar graphs.
But, as we shall now see, when there is a large number of candidates the situation is far worse:
the gerrymandering problem is inapproximable, even in planar graphs!
This is a consequence of the next result.

\begin{restatable}{theorem}{hardplanar}\label{thm:hard-planar}
For planar graphs, finding a partition with one winning district is NP-hard.
\end{restatable}

\begin{proof}
    % This reduction differs from Lemma~\ref{lem:hardness} in that it utilizes more than two colors.
    We construct a reduction from the \textit{minimum connected vertex cover problem}. Here, given a connected graph $G$, we seek a vertex cover $S\subseteq V(G)$ of minimum cardinality such that the subgraph induced by $S$ is connected. This problem 
    is NP-hard even in planar graphs of maximum degree four~\cite{GareyJohnson1977}. 
    Now, given an integer $l$, we build an auxiliary graph $H$ from $G$ such that there is a size $l$ connected vertex cover in $G$ {\em if and only if} there is an $(n,1)$-partition of $H$. 
    The construction of $H$, shown in Figure~\ref{fig:hard-planar}, follows.  
    For every vertex $v \in V(G)$:  
        \begin{itemize}
            \item There is a white vertex $w^v\in V(H)$.
            \item There are $n|E|\cdot (|E|+1)$ red vertices $r^v_i\in V(H)$, for $ i = 1,\dots,n|E|\cdot(|E|+1)$. Each such red vertex has an edge in $E(H)$ to $w^v$.
        \end{itemize}
   For every edge $e=(u,v) \in E(G)$:
   \begin{itemize}
            \item There is a blue vertex $b^e\in V(H)$ with edges to both $w^u$ and $w^v$.
            \item There are $nl\cdot(|E|+1)$ blue vertices $b'^e_i\in V(H)$, for $ i = 1,\dots,nl\cdot(|E|+1)$. Each such blue vertex has an edge in $E(H)$ to $b^e$.
            \item For all $b'^e_i$, $i\in \{ 1,\dots,nl(|E|+1)\}$, there are $|E|$ vertices $c^e_{i,j}$, for $ j = 1, \dots |E|$, each with an edge to $b'^e_i$. All vertices of the form $c^e_{i,j}$ are all colored with the same new color $C^e$.
        \end{itemize}
\begin{figure}[h]
\centerline{\includegraphics[width = 8cm]{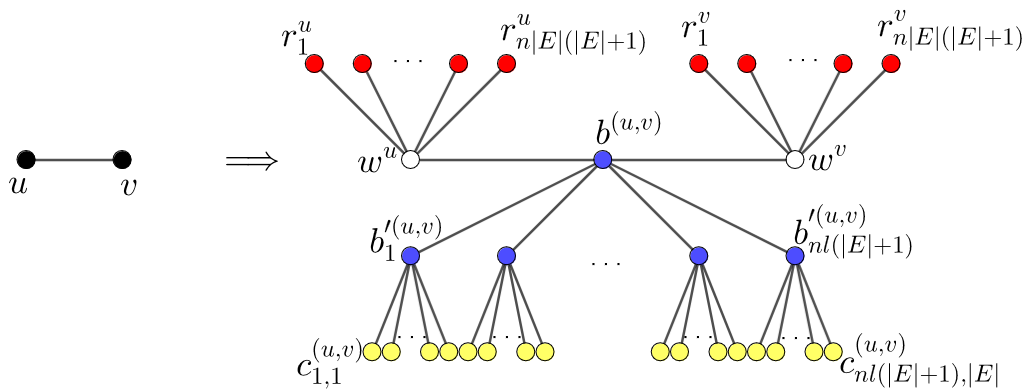}}
\caption{Construction of $H$ in the proof of Theorem~\ref{thm:hard-planar}}\label{fig:hard-planar}
\end{figure}

   Observe that the transformation from $G$ to $H$ preserves planarity. Consider the gerrymandering problem on $H$ with $n$ districts.
   First, we show that if there is a connected vertex cover $S$ in $G$ of size at most $l$ then there is a decomposition of $H$ with one winning district $W$. For all $v \notin S$, we create a district $D_v$ with the vertex $w^v$ and all the vertices of the form $r^v_i$ for $i \in \{1,\dots,n|E|\cdot (|E|+1)\}$. All the remaining vertices are placed in 
   the district $W$. Observe that $W$ is connected because $S$ is a connected vertex cover. In $W$ there are: 
   \begin{itemize}
    \item $|E|\cdot nl\cdot(|E|+1) + |E|$ blue vertices;
    \item $|S|$ white vertices;
    \item $|S| \cdot n|E|\cdot(|E|+1)$ red vertices;
    \item $|E| \cdot nl\cdot(|E|+1)$ vertices of color $C^e$, for each $e \in E$. 
\end{itemize}
So $W$ is winning because it contains more blue than red vertices, as $|S| \leq l$.

For the other direction, assume that $W$ is a winning district in $H$. Define $E^W$ to be the set of 
edges $e\in E(G)$ such that $b^e\in W$ and define $V^W$ to be the set of vertices $v\in V$ such that $w^v\in W$. We claim $V^W$ is a connected vertex cover of cardinality at most $l$ in $G$.
To show this, observe that $V^W$ is connected in $G$ because $W$ is connected in $H$. Next, to show that $V^W$ is a vertex cover, it suffices to prove that $E^W = E$. Suppose, for the sake of contradiction, that $|E^W| < |E|$.  Recall, there are only $n$ districts in total and each district must be connected. So if $b^{(u,v)}\in W$ for some $(u,v)\in E(G)$, then 
at most $n$ vertices of the form $b'^{(u,v)}_i$ are not in $W$ for $i\in \{ 1,\dots,nl(|E|+1)\}$. Thus, at least $(nl\cdot (|E|+1)-n)$ of the vertices $b'^{(u,v)}_i$ are in $W$. By a similar argument, if $b'^{(u,v)}_i \in W$, then at most $n$ vertices of the form $c^e_{i,j}$ are not in $W$, and $(nl(|E|+1)-n)|E|-n$ vertices of $W$ are of the form $c^e_{i,j}$ for $j \in \{1, \dots, |E|\}$. Let $n_{C^e}$ be the number of vertices of $W$ of color $C^e$. Thus, $n_{C^e} \geq (nl(|E|+1)-n)\cdot |E|-n$. Let $n_{blue}$ be the number of blue vertices in $W$. Then:
\begin{align*}
n_{blue} 
&\ \leq\ nl(|E|+1)\cdot |E^W| + |E^W| 
\\
&\leq\ \quad nl(|E|+1)\cdot (|E|-1) + (|E|-1)\\
&\ =\  nl(|E|+1)\cdot |E| - nl(|E|+1) + |E|-1\\
&<\  nl(|E|+1)\cdot |E|-n(|E|+1)\\
&\ \leq\  n_{C^e}
\end{align*}
The strict inequality in the last but one step follows from the assumption that $l\geq 2$.
Therefore, there are more vertices of color $C^e$ than color blue in $W$. This contradicts the assumption that $W$ is a winning district for blue. It follows that $E^W = E$ and $V^W$ is a vertex cover.

It remains to prove that $|V^W| \leq l$. Suppose, for a contradiction, that $|V^W| \geq l+1$. Then, the number $n_{red}$ of red vertices in $W$ is at least $n|E|(|E|+1)\cdot (l+1)-n$. On the other hand, there are in total $nl|E|\cdot(|E|+1)$ blue vertices in $H$. Thus, we have $n_{red} \geq n|E|(|E|+1)(l+1)-n > nl|E|\cdot(|E|+1) \geq n_{blue}$, a contradiction. Therefore, $V^W$ is a connected vertex cover of size at most $l$.
 \end{proof}

\begin{corollary}\label{cor:hard-planar}
Gerrymandering is inapproximable in planar graphs, unless $P=NP$.
\end{corollary}
\begin{proof}
Any approximation algorithm with a finite approximation guarantee must find at least one
winning district if the maximum number of winning districts is strictly positive.
By Theorem~\ref{thm:hard-planar}, this algorithm can then be used to distinguish between
yes and no instances of the connected vertex cover problem in polynomial time.
 \end{proof}

\section{A Constant Factor Approximation for Planar Graphs}\label{sec:constant-approx}

Given our hardness bounds apply even in the setting of unweighted graphs,
it is natural to instigate the study of approximation algorithms for planar graphs in the 
basic case of unweighted graphs.
(We remark that unweighted graphs have been studied in their own right in~\cite{bentert2021complexity,ItoKKO19}).
At first glance, Corollary~\ref{cor:hard-planar} appears fatal to this endeavour: apparently, 
no approximation algorithms exist. However, this conclusion arises only because of the difficultly in 
distinguishing between cases where no districts can be won and cases where at least one district 
can be won. But this implies a large inapproximability bound only applies when the optimal 
number of winning districts $w_{OPT}$ is small.
Indeed our main result is the following.
\begin{theorem}\label{thm:planarapproxlinear}There exists an algorithm with running time $O(n \log{n})$ that computes a $(k,\lfloor w_{OPT}/O(c)\rfloor)$-partition in planar graphs.
\end{theorem}
We remark that the floor function in Theorem~\ref{thm:planarapproxlinear} is necessary
in view of Corollary~\ref{cor:hard-planar}.
An important consequence of the theorem is that there is an $O(c)$-approximation algorithm
for gerrymandering in unweighted planar graphs, provided the optimal number of winning districts 
$w_{OPT}$ is a sufficiently large multiple of $c$. % $w_{OPT}=\Omega(c)$.
% {\bf [need the constant to use here..]}
So for a constant number of candidates,
there is a constant approximation algorithm given the optimal value is large enough.

We will now prove Theorem~\ref{thm:planarapproxlinear}. Towards this aim, we first reduce our problem to the case where winning districts are singletons:
%\subsection{A Reduction to Singleton Winning Districts}
given a $(k, w_{OPT})$-partition there
is a feasible partition in which blue wins a large number of \textit{singleton districts}.
Specifically, the next lemma implies that insisting all winning districts are \textit{singleton} blue vertices will only incur a constant factor loss. 
% As the proof of this lemma closely resembles other proofs in the section, we have included it in Appendix~\ref{sec:missingsingletonproof}.

\begin{restatable}{lemma}{sreduction}\label{lem:singletonsreduction}
    There is a $(k,\lfloor w_{OPT}/(2c+2) \rfloor)$-partition with singleton winning districts.
\end{restatable}
\begin{proof}
    Let $G$ be a planar graph and consider its optimal $(k,w_{OPT})$-partition. Let $W$ be the index set of the districts won by the blue candidate and let $\{D_i: i\in W\}$ be the corresponding set of winning districts in the $(k,w_{OPT})$-partition. Let $B_i\subseteq D_i$ and $B_W = \bigcup_{i\in W} B_i$ be the set of blue vertices in the winning districts. For each $i\in W$ we have $|B_i| \ge |D_i|/c$ because blue is the most preferred of the candidates. Hence
   % \begin{equation}\label{ineq1}
    $\frac{1}{c}\cdot \sum_{i\in W}|D_i| \le  \sum_{i\in W}|B_i|= |B_W|$.
   % \end{equation} 
We now create an auxiliary planar graph $\widetilde{G}$ as follows.
For each losing district, contract the vertices in that district into a single vertex. This can be done as, by definition, each district induces a connected subgraph. Without loss of generality, we may color the resultant losing singleton vertices red. Since $\sum_{i\in W}|D_i| \le c\cdot |B_W|$, we have
    \begin{equation}\label{ineq2}
        \frac{w_{OPT}}{k}
        \ =\ \frac{\sum_{i\in W} 1}{|\widetilde{V}|-\sum_{i\in W} (|D_i|-1)} 
        \ \le\ \frac{\sum_{i\in W}|D_i|}{|\widetilde{V}|} 
        \ \le\ \frac{c\cdot |B_W|}{|\widetilde{V}|}
    \end{equation}
%Consequently, $|\widetilde{V}| \le \frac{ck|B_W|}{w}$. 
As $\Tilde{G}$ is connected it contains a spanning tree $T$.
Let $deg_{T}(v)$ denote the degree of vertex $v$ in the spanning tree $T$.
It holds that:
    \begin{equation*}
        \sum_{v\in B_W} deg_{T}(v) 
        \ \le\ \sum_{v\in \widetilde{V}} deg_{T}(v)
        \ <\ 2 |\widetilde{V}|
        \ \le\ \frac{2ck\cdot |B_W|}{w_{OPT}}
    \end{equation*}
Here the final inequality follows from~(\ref{ineq2}). Consequently,
$\frac{1}{|B_W|}\cdot\sum_{v\in B_W} deg_{T}(v) \ <\ \frac{2ck}{w_{OPT}}$.
Thus, the average degree in the tree $T$ of the blue vertices in the winning district is less than $\frac{2ck}{w_{OPT}}$. Define $B'_W = \{v\in B_W: deg_{T}(v)\le (1+\frac{1}{2c})\cdot\frac{2ck}{w_{OPT}}\}$. By Markov's inequality, we have $|B'_W|\ge \frac{1+\frac{1}{2c}-1}{1+\frac{1}{2c}}\cdot |B_W| = \frac{|B_W|}{2c+1}$. 
Now take a maximal set $B^*\subseteq B'_W$ such that $G\setminus B^*$ contains at most $k-|B^*|$ components. If $B^*=B'_W$ then selecting $B^*$ as
the winning singleton districts gives a factor $2c+1$ approximation guarantee.

\begin{sloppypar}
    Otherwise observe that removing any vertex $v\in B^*$ creates at most 
$\left\lfloor(1+\frac{1}{2c})\cdot\frac{2ck}{w_{OPT}}\right\rfloor = \left\lfloor\frac{k(2c+1)}{w_{OPT}}\right\rfloor$ new components (plus the singleton district $v$ itself). Therefore, by the maximality of $B^*$, we have
$(|B^*|+1)\cdot \left(\left\lfloor\frac{k(2c+1)}{w_{OPT}}\right\rfloor +1\right) \ > \ k$.
Rearranging, we have
\end{sloppypar}
 \begin{equation*}
    |B^*|+1 \ >\ \frac{k}{\left\lfloor\frac{k(2c+1)}{w_{OPT}}\right\rfloor +1}
   \ \ge \ \frac{k}{\frac{k(2c+1)}{w_{OPT}} +1}
    \ = \ \frac{w_{OPT}}{2c+1 +\frac{w_{OPT}}{k}}
    \end{equation*}
   Because $|B^*|+1$ is integral, the strict inequality yields
$|B^*|\ \ge \ \left\lfloor\frac{w_{OPT}}{ 2c+1 +\frac{w_{OPT}}{k}}\right\rfloor$.
So the vertices of $|B^*|$ form the winning districts in a $(\hat{k},\lfloor w_{OPT}/( 2c+2) \rfloor)$-partition with $\hat{k} \leq k$. We can convert the partition obtained from the previous algorithm into a $(k, \lfloor w_{OPT}/(2c+2) \rfloor)$-partition by splitting the losing districts. For each losing district that is not a singleton, we can construct a spanning tree on the districts and remove the leaves of the tree one by one. Each time we remove a leaf, we create a new losing district. We repeat this process until we have obtained $k$ districts in total. This procedure is feasible because $k \leq n$, and thus, we are guaranteed to obtain $k$ districts before all the losing districts become singletons.
 \end{proof}

Theorem~\ref{thm:planarapproxlinear} will follow immediately from Lemma~\ref{lem:singletonsreduction} and the next lemma.
\begin{lemma}\label{lem:planarapproxsingletons}
    In planar graphs, there is an algorithm with running time $O(n \log{n})$ that computes a $(k,\lfloor w_{OPT}/845)\rfloor)$-partition where the winning districts are singletons.
\end{lemma}
So it suffices to prove Lemma~\ref{lem:planarapproxsingletons} and we devote the remainder of this section to its proof. We remark that, for clarity and because the optimized constant would be large regardless, we have chosen not to optimize the constants in this proof. 
The restriction to the singleton winning districts case has some useful implications. In particular,
we may assume without loss of generality that $c=2$. This is because districts that are not  blue singletons are considered losing no matter how many candidates there are. The following algorithm requires that we approximately know $w_{OPT}$, the maximum number winning singletons in a partition of $G$ into $k$ districts. Of course this number is unknown, but we can run the algorithm by guessing $\log{n}$ possible values of $w_{OPT}$ (namely, powers of two) and selecting, from the different partitions obtained, the one that maximizes the number of winning districts. This trick costs us at most a factor~$2$ in the approximation guarantee.

Here is a quick overview of the approximation algorithm which consists of four phases. The algorithm aims to select the lowest degree blue vertex in a greedy manner and convert it into a winning singleton. However, to make this algorithm work effectively, we must first modify the graph. In the initial \textit{pruning} step, we contract the red vertices that are connected and put them in the same losing district. Additionally, we convert blue vertices incident to many red leaves into red vertices since they are poor choices for a winning district. In the subsequent \textit{cut-and-connect} step, we eliminate degree two and three red vertices and replace them with edges and triangles to maintain the graph structure. Although these vertices can significantly increase the degrees of blue vertices, they are not problematic since we can avoid making them losing singletons. Once these modifications are complete, we apply the \textit{greedy} step to the resulting graph by selecting blue vertices with low degree to form singleton winning districts, as long as we do not generate more than $k$ components. Finally, in the \textit{5-color} step, the algorithm selects an independent set from the set of blue singleton winning vertices to ensure that the red vertices that were removed in the \textit{cut-and-connect} step do not become singletons.
Let's now detail these four phases.

\

\noindent{\bf Phase I: Pruning.}
The first step is the following pruning procedure which can be implemented in linear time.

\begin{algorithm}[H]
\caption{Pruning}\label{alg:pruning}
\begin{algorithmic}[1]
\State{\textbf{Input:} a connected, planar graph $G$ with vertices colored blue or red}
\State{Contract all connected red vertices}
\While{$\exists$ a blue vertex incident to more than $\frac{12k}{w_{OPT}}$ red leaves (degree $1$ red vertices in $G$)}
\State{Color it red}
\State{Contract its red component into one red vertex}
\EndWhile
\State{\textbf{Output:} The resultant graph $G_1$}
\end{algorithmic}
\end{algorithm}
\begin{lemma}\label{lem:pruning}
    The graph $G_1$ output by the pruning procedure has the following properties:
    \begin{enumerate}
        \item\label{it:lem-prunning1} The set $R_1$ of red vertices in $G_1$ forms an independent set.
        \item\label{it:lem-prunning2} Every blue vertex in $G_1$ is incident to less than $\frac{12k}{w_{OPT}}$ red leaves.
        \item\label{it:lem-prunning3} Any partition of $G_1$ into $k$ districts induces a partition of $G$ into $k$ districts with the same number of winning blue singletons.
        \item The set $B_1$ of the blue vertices of $G_1$ has cardinality $|B_1| \geq w_{OPT}/2$.
    \end{enumerate}
\end{lemma}
\begin{proof}
Since two red vertices connected by an edge are contracted  in $G_1$, we have that \eqref{it:lem-prunning1} holds. Property \eqref{it:lem-prunning2} holds due to line $3$ of the algorithm. The only graph operation performed to obtain $G_1$ from $G$ is contraction, hence, the connectivity of each district is maintained in $G$. So \eqref{it:lem-prunning3} holds.
% Since the first three properties point are clear, it suffices to
Next, we prove the fourth property. Let $G_{OPT} = (V_{OPT} = B_{OPT}\cup R_{OPT}, E_{OPT}) $ be the graph obtained from the optimal decomposition of $G$ into $k$ districts, after the contraction of each losing district into a singleton red vertex. Observe that $G_{OPT}$ is planar because it is a minor of $G$. Therefore, we have 
    $\sum\limits_{v\in V_{OPT}}deg_{G_{OPT}}(v) \ <\ 6\cdot |V_{OPT}| \ =\ 6k$.
It follows that,
    $\frac{1}{w_{OPT}}\cdot \sum\limits_{v\in B_{G_{OPT}}}deg_{OPT}(v) \ < \ \frac{6k}{w_{OPT}}$.
Therefore, the average degree of blue vertices in $G_{OPT}$ is less than $\frac{6k}{w_{OPT}}$. Define $B'_{OPT}:=\{v\in B_{OPT}: deg_{G_{OPT}}(v) \leq \frac{12k}{w_{OPT}}\}$. By Markov's inequality, we have $|B'_{OPT}| \geq  \frac12\cdot w_{OPT}$. 

We will now prove that every vertex in $B'_{OPT}$ remains blue upon termination of the pruning procedure (Algorithm~\ref{alg:pruning}). Observe that a blue vertex can only turn red (at line $4$ in the algorithm) if it is incident to more than $\frac{12k}{w_{OPT}}$ red leaves at some point. For a contradiction, assume this is true for some vertex $v\in B_{OPT}'$. By definition of $B'_{OPT}$, we have $deg_{G_{OPT}}(v) \leq \frac{12k}{w_{OPT}}$. Therefore at the time when $v$ is incident to more than $\frac{12k}{w_{OPT}}$ red leaves, at least two of those leaves, say $r_1$ and $r_2$ belong to the same losing district in $OPT$. Every district is connected so we can find a path $P$ connecting those two leaves that does not pass through $v$. We have two case. First, at this time, all the vertices on $P$ are red. But then $r_1$ and $r_2$ should have been contracted together, a contradiction. Second, there is a blue vertex $b$ on $P$. But then $r_1$ and $r_2$ cannot be leaves at this time, a contradiction.
Thus we have $|B_1| \geq |B'_{OPT}| \geq  \frac12\cdot w_{OPT}$.
 \end{proof}

\noindent{\bf Phase II: Cut \& Connect.}
The second step of the approximation algorithm is the cut and connect procedure.
This procedure simply (i) replaces each red vertex of degree two by an edge connecting its two blue neighbours, and (ii) replaces each red vertex of degree three by a triangle connecting its three blue neighbours. This is illustrated
in Figure~\ref{fig:cut-and-connect}.

\vspace{-0.3cm}

\begin{algorithm}[H]
\caption{\textit{Cut and Connect}}\label{alg:CandC}
\begin{algorithmic}[1]
\State{\textbf{Input:} $G_1$ from Algorithm 1}
\State{Obtain a new graph $G_2$ from $G_1$ by turning all degree 2 and 3 red vertices into edges that connect the neighbouring blue vertices}
\State{\textbf{Output:} a new graph $G_2$}
\end{algorithmic}
\end{algorithm}

\vspace{-1cm}

\begin{figure}[h]
\centerline{\includegraphics[width = 8cm]{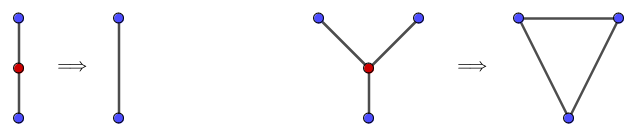}}
\caption{The Cut and Connect procedure}\label{fig:cut-and-connect}\vspace{-.7cm}
\end{figure}

\begin{lemma}\label{lem:CandC}
    The graph $G_2$ output by the cut-and-connect procedure has the properties:
  \begin{enumerate}
        \item The set $R_2$ of red vertices in $G_2$ form an independent set.
        \item Every blue vertex in $G_2$ is incident to less than $\frac{12k}{w_{OPT}}$ red leaves.
        \item The set $B_2$ of blue vertices in $G_2$ has cardinality $|B_2| \geq w_{OPT}/2$.
        \item $G_2$ is planar.
        \item $|B_2| \geq \frac{|V_2|}{2+\frac{12k}{w_{OPT}}}$. 
    \end{enumerate}
\end{lemma}
\begin{proof}
The first three properties hold by Lemma~\ref{lem:pruning}. The fourth property holds because replacing each red vertex of degree two by an edge 
and replacing each red vertex of degree three by a triangle maintains planarity. So let's prove the fifth property that $|B_2| \geq \frac{|V_2|}{2+\tfrac{12k}{w_{OPT}}}$. To show this, we apply one last set of transformations. 
We create a new planar graph $G_2^*$ from $G_2$ as follows:
    \begin{itemize}
        \item Delete all the red leaves of $G_2$.
        \item Delete all the edges between two blue vertices. 
    \end{itemize}
    Clearly $G_2^*=(B_2\cup R^*_2, E^*_2)$ is planar. Furthermore, it is bipartite because both the red vertices $R^*_2$ and the blue vertices $B_2$ form independent sets.  
    It is well-known that the average degree in a planar, bipartite graph is less than four.
   Moreover, because every red vertex in $G_2^*$ has degree at least $4$, the average degree of the blue vertices is less than $4$. Hence:
    \begin{equation}\label{eq:4B>4R}
        4\cdot |B_2| \ > \ \sum_{v\in B_2} deg_{G_2^*}(v) \ =\   \sum_{v\in R_2^*}deg(v) \ \geq\ 4\cdot |R_2^*|
    \end{equation}
    Here the equality holds because $G_2^*$ is bipartite.
   Thus, from (\ref{eq:4B>4R}), we have 
     $|B_2| \ > |R_2^*|$.
    Now consider $G_2$. Its vertex set is $V_2 = R_2^*\cup B_2 \cup \{\text{red leaves of }G_2\}$. Thus
    \begin{align*}
        |V_2| \ =\ |R_2^*| + |B_2| + |\{\text{red leaves of }G_2\}|
        &\leq \ 2\cdot |B_2| +\frac{12k}{w_{OPT}} \cdot |B_2|\\ &= \ \left(2+\frac{12k}{w_{OPT}}\right) \cdot |B_2|
    \end{align*}
    Here the inequality follows from (\ref{eq:4B>4R}) and the second property. 
Rearranging, we obtain
$|B_2| \ge \frac{|V_2|}{2+\tfrac{12k}{w_{OPT}}}$. 
So the fifth property holds. 
 \end{proof}

\noindent{\bf Phase III: Greedy.}
The third step in the approximation algorithm is a greedy procedure.

\vspace{-0.2cm}

\begin{algorithm}[H]
\caption{\textit{Greedy Algorithm on $G_2$}}\label{alg:greedy}
\begin{algorithmic}[1]
\State{\textbf{Input:} $G_2$ from Algorithm 2}
\While{there exists a blue vertex that is not a singleton component}
\State{Pick the blue vertex $v$ with the lowest degree}
\If{the graph obtained by removing the edges adjacent to $v$ has less than $k$ components}
\State{Remove the edges adjacent to $v$ so that $v$ is now a winning district}
\EndIf
\EndWhile
\State{\textbf{Output:} A partition $\{D_1,\dots, D_{\hat{k}}\}$ of the graph $G_2$ where each $D_i$ corresponds to a connected component and $\hat{k} \leq k$}
\end{algorithmic}
\end{algorithm}

%\vspace{-1cm}

\begin{lemma}\label{lem:phase4}
A $(\hat{k},\lfloor \frac{w_{OPT}}{169} \rfloor)$-partition of $G_2$ is output by the greeedy procedure, with $\hat{k} \le k$.
\end{lemma}
\begin{proof}
Again, because $G_2$ is planar, its average degree is less than $6$.
Furthermore, because 
$|B_2| \geq \frac{|V_2|}{2+\tfrac{12k}{w_{OPT}}}$, by property 5 in Lemma~\ref{lem:CandC},
the average degree of the blue vertices in $G_2$ is less than $6\cdot (2+\frac{12k}{w_{OPT}})$. Let $B_{2}':=\{v\in B_{2}: deg_{G_2}(v) \leq  12\cdot (2+\frac{12k}{w_{OPT}})\}$. By Markov's inequality 
and the third property of Lemma~\ref{lem:CandC},
we have $|B_{2}'| \geq  \nicefrac{1}{2}\cdot |B_{2}| \geq \nicefrac{1}{4} \cdot w_{OPT}$.
Now take a maximal set $B^*\subseteq B'_2$ such that $G\setminus B^*$ (the graph obtained by deleting $B^*$ from $G$) contains at most $k-|B^*|$ components. If $B^*=B'_2$ then selecting $B^*$ as
the winning singleton districts gives a factor $4$ approximation guarantee.
Otherwise, observe that any vertex $v\in B^*$ has degree at most $24 + \frac{144k}{w_{OPT}}$, hence removing $v$ creates at most 
$25+\frac{144k}{w_{OPT}}$ new components (including the singleton district $v$ itself). By the maximality of $B^*$,
adding one more element would increase the number of district above $k$. So the number of winning singletons is  at least $\Big\lfloor \frac{k}{25+\tfrac{144k}{w_{OPT}}} \Big\rfloor \geq \Big\lfloor \frac{k}{\tfrac{169k}{w_{OPT}}} \Big\rfloor = \lfloor \frac{w_{OPT}}{169} \rfloor$, where the inequality comes from the fact that $w_{OPT} \leq k$.
 \end{proof}

\noindent{\bf Phase IV: 5-Coloring.}
However, we are not done: we have a partition $\mathcal{D}=\{D_1,\dots, D_k\}$ for $G_2$ but do not yet
have a feasible partition for $G_1$ (and thus $G$).
The issue is we deleted a collection of red vertices 
during the cut-and-connect procedure.
We cannot add them to $\mathcal{D}$ as singleton districts as this will 
create too many districts. So we must connect these red vertices to existing
districts in $\mathcal{D}$. But then the problem is these red vertices were only connected to blue vertices, so adding them to a winning singleton district will cause the district to be losing. 

To overcome this observe that such red vertices had degree two or three
and that, as shown in Figure~\ref{fig:cut-and-connect}, after deletion their neighbours become pairwise adjacent in $G_2$.
Thus, if we choose a set of winning singleton districts that
form an independent set then every deleted red vertex has at least one
neighbour in $G_2$ that is in a losing district. We may then simply
add the red vertex to that losing district.
But can we quickly find a large set winning singleton districts that
form an independent set? Yes, we simply
take advantage of a significant result established in \cite{Chiba1981}: a planar graph can be $5$-colored in linear time. 
This leads to the fourth and final step in our approximation algorithm, the following 5-color procedure. 

\vspace{-0.3cm}

\begin{algorithm}[H]\label{alg:5-Col}
\caption{\textit{5-Color}}
\begin{algorithmic}[1]
\State{\textbf{Input:} $G_2$ and its partition from Algorithm 2}
\State{5-color the graph $G_2$}
\State{Take the vertices from the color class containing the largest number of winning singleton districts as the new winning districts}
\For{each red vertex deleted in the Algorithm \ref{alg:CandC}}
\State{Add it to an adjacent losing district}
\EndFor
\State{\textbf{Output:} a new partition $D'_1,..., D'_{\Tilde{k}}$ of the graph $G_1$ with $\Tilde{k}\leq k$}
\end{algorithmic}
\end{algorithm}
\begin{proof}(of Lemma~\ref{lem:planarapproxsingletons})
Using the largest color class in the $5$-coloring may cost
an extra factor~$5$ in the approximation guarantee beyond that given by Lemma~\ref{lem:phase4}.
Thus the final partition is a $(\Tilde{k},\lfloor \frac{w_{OPT}}{845} \rfloor)$-partition in $G_1$. By Lemma~\ref{lem:pruning}, it induces a $(\Tilde{k},\lfloor \frac{w_{OPT}}{845} \rfloor)$-partition in $G$. We can then turn it into a $(k,\lfloor \frac{w_{OPT}}{845} \rfloor)$-partition by splitting the losing districts as we did in the proof of Lemma~\ref{lem:singletonsreduction}. This concludes the proof of Lemma~\ref{lem:planarapproxsingletons}
and, thus, also of Theorem~\ref{thm:planarapproxlinear}.
 \end{proof}

\section{A PTAS for Singleton Winning Districts}\label{sec:ptas}
By Theorem~\ref{thm:planarapproxlinear}, there is a constant
factor approximation algorithm for the gerrymandering problem in planar graphs.
This result holds for any number of candidates $c$, provided $w_{OPT}$ is a sufficiently large multiple of $c$.
But can we do any better? Is it possible that there is a polynomial time
approximation scheme (PTAS) for the gerrymandering problem in planar graphs
when $w_{OPT}$ is a sufficiently large? This problem remains open. However, we can obtain a PTAS for 
a slight variant of the gerrymandering problem: the {\em singleton winning district gerrymandering problem}. This is the same as the standard gerrymandering problem except the objective is to maximize the number of {\em singleton} winning districts for the blue player. While this objective is contrived, from
a theoretical perspective this is an important problem due to its close relation to the standard gerrymandering problem, as we illustrated by Lemma~\ref{lem:singletonsreduction}.

\begin{restatable}{theorem}{ptas}\label{thm:ptas}
    There is a polynomial-time approximation scheme (PTAS) for the singleton winning district gerrymandering problem in planar graphs.
\end{restatable}

The first step in our PTAS is a preprocessing step 
based upon Baker's classical partitioning technique
for planar graphs~\cite{Baker1994}.

\vspace{-0.5cm}

\begin{algorithm}[H]
\caption{Preprocessing via Baker's Technique}\label{alg:ptas}
\begin{algorithmic}[1]
\State{\textbf{Input:} A planar graph $G$ with vertices colored blue or red, an integer $\lambda$, and an integer $j \in \{0,\lambda-1\}$}
\State{Fix a planar embedding $P$ of $G$}
\State{Set $i = 1$}
\While{$P$ non-empty}
\State{Set $L_i$ to be the vertices on the outer face of $P$}
\If{$i \equiv j \mod \lambda$}
\State{color every vertex $u \in L_i$ red}
\State{Contract all edges $(u,v)$ with  $u,v \in L_i$}
\EndIf
\State{Delete $L_i$ from $P$}
\State{Set $i = i+1$}
\EndWhile
\State{\textbf{Output:} The resultant graph $G_{\lambda,j}$}
\end{algorithmic}
\end{algorithm}

\begin{sloppypar}
So, for a fixed $\lambda$, the preprocessing step creates a set of graphs
$\{G_{\lambda,0}, G_{\lambda,1}, \dots, G_{\lambda,\lambda-1}\}$.
We claim that at least one of these graphs has a solution almost as large
as the optimal solution in the original graph $G$.
This is guaranteed by the next lemma.
\end{sloppypar}
\begin{lemma}\label{lem:preprocess}
With only singleton winning districts the following hold:
\begin{enumerate}
    \item Any $(k,w)$-partition of $G_{\lambda,j}$ induces a $(k,w)$-partition of $G$, for all $k,w$. \label{preprocess1}
    \item For any $\lambda$, there exists $j \in \{0,\lambda-1\}$ and $w^*\geq w_{OPT} \cdot \frac{\lambda-1}{\lambda}$ such that there exists a $(k,w^*)$-partition of $G_{\lambda,j}$. \label{preprocess2}
\end{enumerate}
\end{lemma}
\begin{proof}
Observe that \eqref{preprocess1} holds trivially by construction of the preprocessing algorithm.
So consider \eqref{preprocess2} Let $\lambda$ be an integer and take a $(k,w_{OPT})$-partition of $G$
for the singleton winning district gerrymandering problem. Denote by $B_W$ the set of winning blue vertices. Then there exists $j$ such that $|B_W \cap \bigcup_{i \equiv j \mod{\lambda}}L_i| \leq \frac{1}{\lambda} \cdot w_{OPT}$. Let $w^*:= |B_W \setminus{ \bigcup_{i \equiv j \mod{\lambda}}L_i}|$, the partition in $G$ then induces a $(k,w^*)$-partition in $G_{\lambda,j}$, as required.
 \end{proof}

Lemma~\ref{lem:preprocess} implies that by solving the 
singleton winning district gerrymandering problem in the appropriate 
graph $G_{\lambda,j}$ may produce a PTAS for the original graph $G$.
But how can we solve the problem in $G_{\lambda,j}$.
The key is to demonstrate that each graph $G_{\lambda,j}$ has bounded tree-width. 

\begin{lemma}\label{lem:bounded-treewidth}
 Each graph $G_{\lambda,j}$ has treewidth at most $3\lambda-1$. 
\end{lemma}
\begin{proof}
The proof will use the that a $\lambda$-outerplanar has treewidth at most $3\lambda-1$, as shown by Bodlaender~\cite{Bodlaender1998}. 
Unfortunately we cannot apply Bodlaender's result directly to $G_{\lambda,j}$ because it may have an unbounded outerplanarity index.
However, we can still utilize the result.

Specifically, we prove the lemma by induction on the size of $S=\{i:L_i \neq \emptyset, \, i \equiv j\mod\lambda\}$. For the base case, if $|S|=0$,
then the graph $G_{\lambda,j}$ has at most $\lambda$ layers so
its is $\lambda$-outerplanar and thus has treewidth at most $3\lambda-1$.

For the induction step, consider $\bigcup_{i\leq j}L_i$ and 
$G_{\lambda,j}\setminus\bigcup_{i<j}L_i$. The former has treewidth
at most $3\lambda-1$ because it is $\lambda$-outerplanar. The latter 
has treewidth at most $3\lambda-1$ by the induction hypothesis.
To show that $G_{\lambda,j}$ has has treewidth at most $3\lambda-1$, 
observe that the outer face of a connected planar graph is itself connected.
Now consider a contraction step in the preprocessing algorithm.
This implies that each vertex in $L_j$ that remains after the contraction step must correspond to the outer face of a connected component in the graph $G\setminus\bigcup_{l<i}L_l$. 
Now let $T$ be tree decomposition $\bigcup_{i\leq j}L_i$ after the vertices in $L_j$ have been contracted.
Further, for each vertex $x\in L_j$, let $T_x$ be a tree decomposition of the corresponding connected component of $G_{\lambda,j}\setminus\bigcup_{i<j}L_i$. We may root each $T_x$ at a node containing $x$ in its bag. We can then merge $T$ with $T_x$ by connecting the root of $T_x$ to any node of $T$ that contains $x$ in its bag. We do this for each $x\in L_j$. This gives a tree
decomposition of width $3\lambda-1$. To see that the resultant tree decomposition is valid observe that the paths in $G_{\lambda,j}$ connecting a vertex of $T$ to a vertex of $T_x$ must necessarily pass through the vertex $x$, which corresponds to the outer face of the connected component induced by the verticess in $T_x$.
 \end{proof}

\begin{proof}[Proof of Theorem~\ref{thm:ptas}]
Let $\epsilon > 0$. We choose $\lambda = \lceil \frac{1+\epsilon}{\epsilon} \rceil$. For each $j$, we run the dynamic programming algorithm described in Section~\ref{section:boundedTW} on $G_{\lambda,j}$. Since $G_{\lambda,j}$ has bounded treewidth by Lemma~\ref{lem:bounded-treewidth} and Theorem~\ref{thm:boundedTW}, this can be done in polynomial time. Now using Lemma~\ref{lem:preprocess}, there exists a $j$ such that the algorithm produces a $(k,w\cdot \frac{\lambda-1}{\lambda})$-partition of $G_{\lambda,j}$ inducing a $(k,w\cdot \frac{\lambda-1}{\lambda})$-partition of $G$. We picked $\lambda$ such that $\frac{\lambda-1}{\lambda} \geq \frac{1}{1+\epsilon}$ so we get a $(1+\epsilon)$-approximation. Since this holds for any $\epsilon > 0$, we have a polynomial-time approximation scheme.
 \end{proof}

We can combine Theorem~\ref{thm:ptas} with Lemma~\ref{lem:singletonsreduction} and get the following corollary.
\begin{corollary}\label{cor:planarapproxscheme}
    For any $\epsilon > 0$, there exists an algorithm that can compute a $(k,\lfloor \frac{w_{OPT}}{(1+\epsilon)(2c+2)}\rfloor)$-partition for planar graphs.
\end{corollary}
We remark, that while the approximation constant obtained above is better than the one of Theorem~\ref{thm:planarapproxlinear}, the running time is significantly worse.

\section{Conclusion}\label{sec:conclusion}
Several interesting open problems remain.
First, we proved that finding one winning district in a planar graph is NP-hard.
But this reduction required a large number of candidates. 
If the number of candidates is a constant, 
is there a polytime algorithm to find one winning district?
Second, we presented approximation algorithms for planar graphs, given $w_{OPT}$ is sufficiently large. Our hardness result proves that this condition is necessary, but can the approximation guarantee be improved? Specifically, is there a polynomial time approximation scheme (PTAS) for the gerrymandering problem in planar graphs when $w_{OPT}$ is sufficiently large? 
% We achieved an exact algorithm for $\lambda$-outerplanar graphs.
%
Third, as a special case, our approximation algorithms provided solutions where the winning districts are singletons. 
In most natural applications, the sizes of the districts are expected to be similar. Is there an approximation algorithm for planar graphs when the sizes of the districts are (roughly) equal?

%Finally, we suggest a direction for future work, which is to change the optimization problem to be closer to the global election model. Specifically, is it possible to approximate the optimal ratio between the number of winning districts of the target candidate to that of each opponent, instead of maximizing the number of blue winning districts?

 \newpage
%%%%%%%%%%%%%%%%%%%%%%%%%%%%%%%%%%%%%%%%%%%%%%%%%%%%%%%%%%%%%%%%%%%%%%%%

%%%%%%%%%%%%%%%%%%%%%%%%%%%%%%%%%%%%%%%%%%%%%%%%%%%%%%%%%%%%%%%%%%%%%%%%

%%% The acknowledgments section is defined using the "acks" environment
%%% (rather than an unnumbered section). The use of this environment 
%%% ensures the proper identification of the section in the article 
%%% metadata as well as the consistent spelling of the heading.

\section*{Ethics Statement}
This research is purely theoretical and a long term aim of this line of research is to understand, recognize and prevent abuses of electoral systems. This paper strictly upholds the principles of research ethics, research integrity, and social responsibility. Therefore, the findings presented in this paper do not encompass any elements that could potentially have adverse societal impacts.

%%%%%%%%%%%%%%%%%%%%%%%%%%%%%%%%%%%%%%%%%%%%%%%%%%%%%%%%%%%%%%%%%%%%%%%%

%%% The next two lines define, first, the bibliography style to be 
%%% applied, and, second, the bibliography file to be used.

\bibliographystyle{ACM-Reference-Format} 
\bibliography{AAMASbib}

%%%%%%%%%%%%%%%%%%%%%%%%%%%%%%%%%%%%%%%%%%%%%%%%%%%%%%%%%%%%%%%%%%%%%%%%

\clearpage
\appendix
\section{A Dynamic Program for Graphs of Bounded Treewidth}\label{sec:missingDPproof}

In this appendix, we present a dynamic program for the gerrymandering problem in graphs of bounded treewidth. In particular, we prove Theorem~\ref{thm:boundedTW}.

\boundedTW*
We begin with the necessary definitions concerning treewidth.
\begin{definition}\label{def:treewidth}
A \emph{tree decomposition} of a graph $G$ is a pair $(T,\beta)$ of a tree $T$
and $\beta:V(T) \rightarrow 2^{V(G)}$, 
such that:
\begin{itemize}
    \item $\bigcup\limits_{v \in V(T)} \beta(v) = V(G)$, and
    \item\label{item:twedge} for any edge $\{\rho,\rho'\} \in E(G)$ there exists a node $v \in V(T)$ such that $\rho,\rho' \in \beta(v)$, and
    \item\label{item:twconnected} for any vertex $\rho \in V(G)$, the subgraph of $T$ induced by the set $T_\rho = \{v\in V(T): \rho\in\beta(v)\}$ is a tree.
\end{itemize}
The {\em width} of $(T,\beta)$ is $\max\limits_{v\in V(T)}\{|\beta(v)|\}-1$. The {\em treewidth} of $G$ is the minimum width of a tree decomposition of $G$.
\end{definition}
For $v \in V(T)$, we say that $\beta(v)$ is the \emph{bag} of $v$. Let $\tw(G)$ denote the treewidth of $G$. Now if $(T,\beta)$ denotes a tree decomposition of $G$, we can transform in polynomial time $(T,\beta)$ into a {\em nice} tree decomposition of the same width~\cite{DBLP:journals/siamcomp/Bodlaender96} and with $O(n)$ nodes. A nice tree decomposition is one where $T$ is a binary tree, a node with two children has the same bag as its children, and any other node either has an empty bag or a bag that differs by exactly one vertex from the bag of its child. Moreover, the bag of the root, denoted by $r$, is empty. 

Let $G = (V, E)$ be a weighted graph with weight function $wt: V(G) \rightarrow \mathbb{N}$, and let $T$ be a nice tree decomposition of $G$ with treewidth $\tw(G)$. For a node $X$ of the tree decomposition, a \textit{configuration} of $X$ is a tuple $(\kappa,\omega,P,\gamma)$. A configuration of $X$ is \textit{valid} if there exists a partition $\mathbf{P}$ of the subgraph induced by the subtree rooted at $X$ such that:
\begin{itemize}
    \item $\kappa\in \{1,\dots,k\}$ is the number of districts in $\mathbf{P}$
    \item $\omega \in \{1,\dots,\kappa\}$ is the number of districts won by candidate $1$ (the blue candidate).
    \item $P$ is a partition of $\beta(X)$, and two vertices in $\beta(X)$ are in the same district in $P$ if and only if they are in the same district in $\mathbf{P}$.
    \item $\gamma: P \times \{1,\dots,c\}  \rightarrow \{1,\dots,\sum\limits_{v\in V} wt(v)\}$ is a function; $\gamma(D,i)$ indicates for any district $D$ of $X$ and any color $i\in \{1,\dots,c\}$ the number of votes received by candidate $i$ in the district $\mathbf{D}\in \mathbf{P}$ such that $D\subset \mathbf{D}$.
\end{itemize}
In that case we say that $(\kappa,\omega,P,\gamma)$ is \textit{coherent} with $\mathbf{P}$.

The dynamic program calculates the function $\varphi$, which takes a node $X$ of the tree decomposition as input and generates a list of all valid configurations. We employ a bottom-up approach, whereby the dynamic program first processes the leaf nodes of the tree decomposition $T$, and then computes $\varphi({X})$ for all nodes in $T$.

The initialization step is straightforward since all leaf nodes are empty. The updating steps vary depending on the type of node, and the following programs provide explanations for each type.\\

\noindent\textbf{Join Node:} Suppose $X$ is a join node with children $L$ and $R$, and we have already computed $\varphi(L)$ and $\varphi(R)$. Let $(\kappa_L,\omega_L,P_L,\gamma_L)$ and $(\kappa_R,\omega_R,P_R,\gamma_R)$ be two valid configurations of $L$ and $R$, respectively. Then there exist two partitions $\mathbf{P}_L$ and $\mathbf{P}_R$ of the subgraphs induced by the subtrees rooted in $L$ and $R$, respectively, such that $\mathbf{P}_L$ is coherent with $(\kappa_L,\omega_L,P_L,\gamma_L)$ and $\mathbf{P}_R$ is coherent with $(\kappa_R,\omega_R,P_R,\gamma_R)$.

We can merge $\mathbf{P}_L$ and $\mathbf{P}_R$ into a new partition $\mathbf{P}$ if and only if $P_L = P_R$. In this case, the resulting configuration $(\kappa,\omega,P,\gamma)$ coherent with $\mathbf{P}$ is:
\begin{itemize}
    \item We have $P= P_L = P_R$.
    \item A district is either only in $\mathbf{P}_L$, only in $\mathbf{P}_R$ or it intersects $\beta(X)$. Therefore $\kappa = \kappa_L + \kappa_R - |P|$. If $\kappa > k$, the configuration is not valid.
    \item Similarly for any district $D$ of $P$ and any color $i$, we have \[\gamma(D,i) = \gamma_L(D,i) + \gamma_R(D,i) - \sum\limits_{v\in D:v\text{ colored }i} wt(v).\]
    \begin{sloppypar}
    \item To compute the number of winning districts, we first define $W_L$, $W_R$ and $W$ as the sets of districts in $P$ that are won by blue in the corresponding districts of $\mathbf{P}_L$, $\mathbf{P}_R$ and $\mathbf{P}$ respectively. Specifically, $W_L := \{D\in P: \gamma_L(D,1)>\gamma_L(D,i) \text{ for any color }i \neq 1\}$, $W_R := \{D\in P: \gamma_R(D,1)>\gamma_R(D,i) \text{ for any color }i \neq 1\}$ and $W := \{D\in P: \gamma(D,1)>\gamma(D,i) \text{ for any color }i \neq 1\}$. Next, we compute the number of winning districts in $\mathbf{P}_L$ and $\mathbf{P}_R$ that do not intersect $\beta(X)$. Specifically, the number of winning districts in $\mathbf{P}_L$ is $\omega_L - |W_L|$, and the number of winning districts in $\mathbf{P}_R$ is $\omega_R - |W_R|$. Finally, we compute the number of winning districts in the merged partition $\mathbf{P}$: $\omega = \omega_R + \omega_L - |W_L| - |W_R| + |W|$.
    \end{sloppypar}
\end{itemize}

To compute $\varphi(X)$ from $\varphi(L)$ and $\varphi(R)$, we need to consider all possible pairs of $(\kappa_L,\omega_L,P_L,\gamma_L) \in \varphi(L) $ and $(\kappa_R,\omega_R,P_R,\gamma_R)\in \varphi(R)$. For each pair, we check if a valid configuration exists. If so, we add it to a list and remove duplicate entries.\\

\noindent\textbf{Forget Node:} Assume $X$ is a forget node with a child $Y$ such that $\beta(X) = \beta(Y)\setminus {v}$.  Let $(\kappa_Y,\omega_Y, P_Y, \gamma_Y)$ be a valid configuration for $Y$ coherent with a partition $\mathbf{P}$ of the subgraph induced by the subtree rooted in $Y$. Note that the subtree rooted at $X$ induces the same subgraph as the subtree rooted at $Y$. There are two cases to consider:

Case 1: ${v}\in P_Y$ is a singleton district. In this case, the configuration $(\kappa,\omega,P,\gamma)$ of $X$ coherent with $\mathbf{P}$ is given by $\kappa = \kappa_Y$, $\omega = \omega_Y$, $P = P_Y\setminus{{y}} $, and $\gamma = \gamma_Y|P$ where $\gamma_Y|P$ is defined to be the same as $\gamma_Y$ for each district in $P$.

Case 2: $v \in D$ for some district $D\in P_Y$. Then, the configuration coherent with $\mathbf{P}$ is given by:
$\kappa = \kappa_Y$, $\omega = \omega_Y$, $P$ is obtained from $P_Y$ by removing ${v}$ from the district $D$, so $P = (P_Y\setminus {D}) \cup ({D \setminus{v}})$. We have $\gamma(D',i) = \gamma_Y(D',i)$ for any color $i$ and district $D \neq D \setminus{v}$, and $\gamma(D\setminus{v},i) = \gamma_Y(D,i)$ for any color $i$.\\

\noindent\textbf{Introduce Node:} Assume $X$ is an introduce node with child $Y$ such that $\beta(X) = \beta(Y) \cup \{v\}$. Let $(\kappa_Y,\omega_Y, P_Y, \gamma_Y)$ be a valid configuration of $Y$ coherent with a partition $\mathbf{P}$ of the subgraph induced by the subtree rooted in $Y$. Note that then the subgraph induced by the subtree rooted at $X$ is the same as the one rooted at $Y$ with an extra vertex $v$. We need to compute all possible configurations $(\kappa, \omega, P, \gamma)$ obtained by extending $\mathbf{P}$. Define $N(v):=\{u \in \beta(X):(u,v)\in E\}$, the neighborhood of $v$ in the graph induced by $\beta(X)$. For any subset $S\subseteq N(v)$ ($S$ can be empty), the districts induced by the vertices in $S$ can be connected using $v$ and merged. In that case the configuration $(\kappa, \omega, P, \gamma)$ can be computed from $(\kappa_Y,\omega_Y, P_Y, \gamma_Y)$. First we define the set $\Delta_S := \{D\in P_Y : D \cap S \neq \emptyset \}$, the set of districts of $P_Y$ intersecting $S$.
\begin{itemize}
\item The new partition is obtained by merging districts in $\Delta_S$ with $\{v\}$: $$P = P_Y\setminus \Delta_S \bigcup \left\{ (\bigcup\limits_{D\in\Delta_S} D) \cup \{v\}\right\}$$
\item The number of districts is $\kappa = \kappa_Y - |\Delta_S|+1$.
\item For any color and any district $D\in P_Y$ not in $\Delta_S$, we have $\gamma(D,i) = \gamma_Y(D,i).$ For the district $(\bigcup\limits_{D\in\Delta_S} D) \cup \{v\}$, we have \[\gamma((\bigcup\limits_{D\in\Delta_S} D) \cup \{v\},i) = \sum \limits_{D\in \Delta_S} \gamma_Y(D,i) + \begin{cases} wt(v) & \text{ if } v\text{ colored } i \\
0 & \text{ otherwise.}\end{cases}\]
\item To compute the updated number of winning districts, we need to first remove those winning districts in $\Delta_S$ and recompute the winning status after merging $\Delta_S$ with $\{v\}$. The number of winning districts in $\Delta_S$ is given by $|\{D\in \Delta_S: \gamma_Y(D,1)>\gamma_Y(D,i), \forall i \not=1\}|$. The new merged district $(\bigcup\limits_{D\in\Delta_S} D) \cup \{v\}$ is winning if \[\gamma((\bigcup\limits_{D\in\Delta_S} D) \cup \{v\},1) > \gamma((\bigcup\limits_{D\in\Delta_S} D) \cup \{v\},i), \forall i\not=1.\] Concretely:
\begin{multline*}
    \omega = \omega_Y - |\{D\in \Delta_S: \gamma_Y(D,1)>\gamma_Y(D,i), \forall i \not=1\}| \\+ \begin{cases}
    1, &\text{if } \gamma((\bigcup\limits_{D\in\Delta_S} D) \cup \{v\},1) > \gamma((\bigcup\limits_{D\in\Delta_S} D) \cup \{v\},i), \forall i\not=1\\
    0, &\text{otherwise. }
\end{cases}
\end{multline*}
\end{itemize}

Finally, at the root node, the largest $\omega$ is the maximum number of winning districts for the target candidate $1$ satisfying the constraint that the total number of districts is at most~$k$.\\

\noindent\textbf{Time Complexity:} Let $S$ be the number of possible configurations $(\kappa, \omega, P, \gamma)$ of a single bag. For each configuration, there are $k$ possible values for $\kappa$ and $\omega$; the number of partitions is upper bounded by $\tw^{\tw}$; and the number of functions $\gamma: P \cdot \{1,\dots, c\} \rightarrow \{1,\dots, \sum\limits_{v\in V} wt(v)\}$ is at most $(\sum\limits_{v\in V} wt(v))^{c\cdot\tw}$. So $S \le k^2 \cdot \tw^{\tw} \cdot (\sum\limits_{v\in V} wt(v))^{c\cdot\tw} \le n^{2+tw} \cdot (\sum\limits_{v\in V} wt(v))^{c\cdot\tw}$. We now bound the time to compute all the valid configurations of a node, given the valid configurations of its child nodes. 

For a join node, the algorithm will consider all pairs of configurations of its child nodes and there are at most $S^2$ such pairs. For each pair of configuration, it takes constant time to compute $P$ and $\kappa$; and at most $O(c\cdot\tw) = O(n^2)$ time to compute both $\gamma$ and $\omega$, so the running time in this case is bounded by $O(S^2 \cdot n^2)$.  For a forget node, there is only one child with $S$ different configurations to consider. Fixing a configuration, it takes constant time to compute all parameters but $\gamma$, which takes time $O(c\cdot\tw) = O(n^2)$. For a introduce node, there is also only one child. Fixing a configuration, the algorithm needs to compute all possible partitions. The number of possible choice is bounded by $(\tw+1)$. The running time for computing $\omega$ and $\gamma$ is at most $O(c\cdot\tw) = O(n^2)$ each. It follows that the time complexity for computing a single node is dominated by the join nodes. Therefore the total running time is bounded by $O(n)\cdot O(S^2 \cdot n^2) = O(n^{2\tw+7}\cdot (\sum\limits_{v\in V} wt(v))^{2c\cdot\tw})$.\\

This completes the proof of Theorem~\ref{thm:boundedTW}.

\section{Hardness Proofs}\label{sec:missinghardnessproof}

In this appendix, we give the prove our hardness result for gerrymandering in general graphs.

\hardness*
\begin{proof}
    Let $G$ be a graph with $|V(G)|=n$. We derive an auxiliary graph $H$ from $G$ such that there exists an  independent set of size $w$ in $G$ if and only if there exists a $(n,w)$-partition of $H$. The construction of $H$, illustrated in
    Figure~\ref{fig:hard-general}, is as follows:
    \begin{itemize}
        \item For each vertex $v \in V(G)$, there is a blue vertex $b^v\in V(H)$.
        \item For each edge $e=(u,v) \in E(G)$, we create $2n$ red vertices $r^e_i\in V(H)$ for $i = 1,\dots,2n$. Each such red vertex has edges in $E(H)$ to both $b_v$ and $b_u$.
    \end{itemize}
\begin{figure}[h]
\centerline{\includegraphics[width = 8cm]{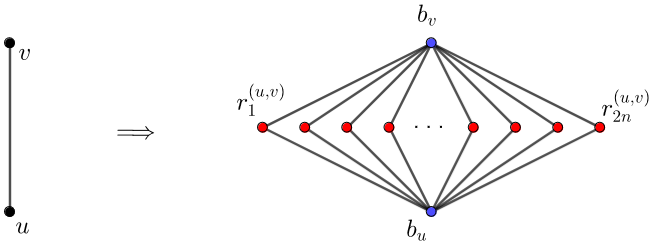}}
\caption{Construction of $H$ in the proof of Lemma~\ref{lem:hardness}}\label{fig:hard-general}
\end{figure}

Now consider the gerrymandering problem in $H$ with $n$ districts. The first task is to prove that in a valid decomposition of $H$ into $n$ connected districts any winning district must be a blue singleton vertex. Suppose we have a winning district $D_i$ that is not a blue singleton. Then $D_i$ must contain at least two blue vertices. 
Furthermore, by the connectivity requirement, there must be vertices $b_u, b_v \in D_i$ such that the edge $e = (u,v)\in E(G)$. Because $H$ only has $n$ blue vertices, $D_i$ can have at most $n-1$ red vertices. Consequently, there must be at 
least $n+1$ red vertices of the form $r_j^e$, for some $j \in \{ 1,\dots,2n\}$, which are {\bf not} in $D_i$. 
But these red vertices must each then form a singleton district. Thus there are at least $n+1$ districts, a contradiction.

We will now demonstrate that the set $\{b^v: v\in S\}$ is a collection of blue winning singleton-districts in $H$ if and only if $S$ is an independent set in $G$. First, assume $S$ is an independent set in $G$. We create a partition of $V(H)$ such that $V(H) = \bigcup_{v\in V(G)} D_v$ and start with $D_v = {b_v}$, for each $v\in V(G)$. 
Now since $S$ is an independent set, $V(G)\setminus S$ is a vertex cover in $G$.
Thus, for every $e\in E(G)$, there exists $v_e \in V(G)\setminus S$ which is an endpoint of $e$. We add all vertices of the form $r^e_i$ to $D_{v_e}$. Observe this induces a partition of $H$ into $n$ connected districts where $\{b^v: v\in S\}$ in $H$ are a collection of blue-winning singleton districts.
Second, assume that $\{b^v: v\in S\}$ is a collection of blue-winning singleton districts in $H$. Assume for a contradiction 
that $S$ is not an independent set in $G$. Then there exists an edge $(u,v)\in E(G)$ where $\{u, v\}\subseteq S$. This implies all the vertices $r_j^e$, for $j \in \{ 1,\dots,2n\}$, must be singleton districts, contradicting the fact there are $n$ districts. 
 \end{proof}

\hardgeneral*

\begin{proof}
According to the result of Zuckerman \cite{Z2006}, the maximum independent set (or clique) problem cannot be approximated within a factor of $n^{1-\epsilon}$ unless $P = NP$, strengthening the work of H\r{a}stad~\cite{Hstad1999}. The theorem then follows from Lemma~\ref{lem:hardness}, because the number of vertices of the graph $H$ used in the proof is bounded by $|V(G)| + 2\cdot |V(G)| \cdot |E(G)|  \leq |V(G)|^3$.
 \end{proof}

\end{document}